\documentclass[11pt,reqno]{amsart}
\usepackage{amsmath,amssymb,amsthm, url,hyperref,cite}
\newtheorem{thm}{Theorem}

\newtheorem{prop}{Proposition}
\theoremstyle{definition}
\newtheorem{defn}{Definition}
\theoremstyle{remark}

\title[Existence theorem for weak quasiperiodic solutions]
{Existence theorem\\ for weak quasiperiodic solutions\\ of Lagrangian systems\\
on Riemannian manifolds}
\author[I. Parasyuk]
{Igor Parasyuk}
\address{National Taras Shevchenko Univesity of
Kyiv\\ Volodymyrs'ka 64, Kyiv, 01033, Ukraine}
\author[A. Rustamova]
{Anna Rustamova} \email{pio@univ.kiev.ua,\;anna\_rustamova@hotmail.com}

\begin{document}

 \begin{abstract}We establish new sufficient conditions for the existence of weak Besicovitch
quasiperiodic solutions for natural Lagrangian system on Riemannian manifold
with time-quasiperiodic force function.\end{abstract}

\maketitle

\textbf{1. Introduction.} Let $\mathcal{M}$ be a smooth complete connected
$m$-dimensional Riemannian manifold equipped with an inner product $\left\langle
\cdot,\cdot\right\rangle $ on fibers $T_{x}\mathcal{M}$ of tangent bundle $T\mathcal{M}$. Consider a
natural system on $\mathcal{M}$  with Lagrangian function of the form
$L\bigl|_{T_{x}\mathcal{M}}=\frac{1}{2}\langle\dot{x},\dot{x}\rangle-\Pi(t,x)$ where
$\frac{1}{2}\langle\dot{x},\dot{x}\rangle$ and $\Pi(t,x)$ stand for kinetic and potential
energy respectively. We suppose that the potential energy is represented as
$\Pi:=-W(\omega t,x)$ where $W(\omega t,x)$ is $\omega$-quasiperiodic force function
generated by a function
$W(\cdot,\cdot)\in\mathrm{C}^{0,2}(\mathbb{T}^{k}\times\mathcal{M}\!\mapsto\!\mathbb{R})$ ($W(\cdot,\cdot)$ is
continuous together with $W_{xx}^{\prime\prime}(\cdot,\cdot)$); here
$\mathbb{T}^{k}=\mathbb{R}^{k}/2\pi\mathbb{Z}^{k}$ is $k$-dimensional torus and
$\omega=(\omega_{1},...,\omega_{k})\in\mathbb{R}^{k}$ is a frequencies vector with
rationally independent components. The problem is to detect in such a system
$\omega$-quasiperiodic oscillations.

J.~ Blot in his series of papers \cite{Blo88,Blo89,Blo89A,Blo93}
applied variational method to establish the existence of weak almost
periodic solutions for systems in $\mathbb{E}^{m}$. Later, this method
was used in \cite{BerZha96,Maw98,ZakPar99,Ort09} to prove the existence
of weak and classical almost periodic solutions for systems of variational
type. In \cite{ZakPar99A,ZakPar99B}, weak and classical quasiperiodic
solutions were found for natural mechanical systems in convex compact
subsets of Riemannian manifolds with non-positive sectional curvature.
The goal of the present paper is to extend these results to natural
systems on arbitrary Riemannian manifolds.

\medskip{}

\textbf{2. Variational method.} One can interpret a natural system on $\mathcal{M}$ as
a natural system in Euclidean space $\mathbb{E}^{n}$ (of appropriate dimension
$n$) with holonomic constraint. Namely, in view of the Nash embedding theorem
\cite{Nash56} we consider $\mathcal{M}$ as a submanifold of $\mathbb{E}^{n}$ for some
natural $n>m$. The set $\mathcal{M}\subset\mathbb{E}^{n}$ play the role of holonomic
constraint for natural system in $\mathbb{E}^{n}$ with kinetic energy
$K=\frac{1}{2}\langle\dot{y},\dot{y}\rangle_{\mathbb{E}^{n}}$ and potential energy $-W(\omega
t,y)$, if we suppose that $W(\cdot,\cdot)$ is defined in
$\mathbb{T}^{k}\times\mathbb{E}^{n}$.

In what follows we shall use identical notations for inner product
$\left\langle \cdot,\cdot\right\rangle_{\mathbb{E}^{n}}$ of $\mathbb{E}^{n}$
and the induced inner product $\left\langle \cdot,\cdot\right\rangle $
on $T\mathcal{M}$. Let $\nabla_{\xi}$ stands for the covariant differentiation
of Levi-Civita connection in the direction of vector $\xi\in T\mathcal{M}$,
and let $\nabla f$ stands for gradient vector field of a scalar function
$f(\cdot):\mathcal{M}\mapsto\mathbb{R}$, i.e $\left\langle \nabla f(x),\xi\right\rangle =\mathrm{d}f(x)(\xi)$
for any $\xi\in T_{x}\mathcal{M}$.

Denote by $\mathrm{H}(\mathbb{T}^{k}\!\mapsto\!\mathbb{E}^{n})$ the
space of $\mathbb{E}^{n}$-valued functions on $k$-torus which are
integrable with the square of Euclidean norm $\left\Vert \cdot\right\Vert :=\sqrt{\left\langle \cdot,\cdot\right\rangle}$.
Define on $\mathrm{H}(\mathbb{T}^{k}\!\mapsto\!\mathbb{E}^{n})$ the
standard scalar product $\langle\cdot,\cdot\rangle_{0}=(2\pi)^{-k}\int_{\mathbb{T}^{k}}\langle\cdot,\cdot\rangle\mathrm{d}\varphi$
and the corresponding semi-norm $\left\Vert \cdot\right\Vert_{0}:=\sqrt{\left\langle \cdot,\cdot\right\rangle_{0}}$.
By $\mathrm{H}_{\omega}^{1}(\mathbb{T}^{k}\!\mapsto\!\mathbb{E}^{n})$
denote the space of functions $f(\cdot)\in\mathrm{H}(\mathbb{T}^{k}\!\mapsto\!\mathbb{E}^{n})$
each of which has weak (Sobolev) derivative $D_{\omega}f(\cdot)\in\mathrm{H}(\mathbb{T}^{k}\!\mapsto\!\mathbb{E}^{n})$
in the direction of vector $\omega$. Recall that a function $u(\cdot)\in\mathrm{H}(\mathbb{T}^{k}\!\mapsto\!\mathbb{E}^{n})$
with Fourier series $\sum_{\mathbf{n}\in\mathbb{Z}^{k}}u_{\mathbf{n}}\mathrm{e}^{\mathrm{i}\mathbf{n}\cdot\varphi}$
has a weak derivative iff the series $\sum_{\mathbf{n}\in\mathbb{Z}^{k}}\left|\mathbf{n}\cdot\omega\right|^{2}\left\Vert u_{\mathbf{n}}\right\Vert^{2}$
converges and then the Fourier series of $D_{\omega}u(\cdot)$ is
$\sum_{\mathbf{n}\in\mathbb{Z}^{k}}\mathrm{i}(\mathbf{n}\cdot\omega)u_{\mathbf{n}}\mathrm{e}^{\mathrm{i}\mathbf{n}\cdot\varphi}$.

The space $\mathrm{H}_{\omega}^{1}(\mathbb{T}^{k}\!\mapsto\!\mathbb{E}^{n})$
is equipped with the semi-norm $\left\Vert \cdot\right\Vert_{1}$
generated by the scalar product $\left\langle D_{\omega}\cdot,D_{\omega}\cdot\right\rangle_{0}+\left\langle \cdot,\cdot\right\rangle_{0}$.
After identification of functions coinciding a.e., both spaces becomes
Hilbert spaces with norms $\left\Vert \cdot\right\Vert_{0}$ and
$\left\Vert \cdot\right\Vert_{1}$ respectively.

To any function $u(\cdot)\in\mathrm{H}(\mathbb{T}^{k}\!\mapsto\!\mathbb{E}^{n})$
with Fourier series $\sum_{\mathbf{n}\in\mathbb{Z}^{k}}u_{\mathbf{n}}\mathrm{e}^{\mathrm{i}\mathbf{n}\cdot\varphi}$,
one can put into correspondence a Besicovitch quasiperiodic function
$x(t)=u(\omega t)$ defined by its Fourier series $\sum_{\mathbf{n}\in\mathbb{Z}^{k}}u_{\mathbf{n}}\mathrm{e}^{\mathrm{i}(\mathbf{n}\cdot\omega)t}$.
If $u(\cdot)\in\mathrm{H}_{\omega}^{1}(\mathbb{T}^{k}\!\mapsto\!\mathbb{E}^{n})$
then $\dot{x}(t)$ denotes a Besicovitch quasiperiodic function $D_{\omega}u(\omega t)$.

We define weak solution of Lagrangian system on $\mathcal{M}$ with
density $L=\frac{1}{2}\left\langle \dot{x},\dot{x}\right\rangle +W(\omega t,x)$
in a slightly different way then in \cite{ZakPar99}. First, for any
bounded subset $\mathcal{A}\subseteq\mathcal{M}$, put \[
\mathcal{S}_{\mathcal{A}}:=\mathrm{C}^{\infty}\left(\mathbb{T}^{k}\!\mapsto\!\mathcal{A}\right).\]
 Observe that if $u_{j}(\cdot)\in\mathcal{S}_{\mathcal{A}}$ is a
sequence bounded in $\mathrm{H}_{\omega}^{1}(\mathbb{T}^{k}\!\mapsto\!\mathbb{E}^{n})$ and
convergent to a function $u(\cdot)$ by norm of the space
$\mathrm{H}(\mathbb{T}^{k}\!\mapsto\!\mathbb{E}^{n})$ (recall that we consider the set
$\mathcal{A}\subseteq\mathcal{M}$ both as a subset of $\mathbb{E}^{n}$), then for any
$\mathbf{n}\in\mathbb{Z}^{k}$ the sequence of Fourier series coefficients
$u_{\mathbf{n},j}$ converges to $u_{\mathbf{n}}$ and for some $K>0$ we have \begin{gather*}
\sum_{\left|\mathbf{n}\right|\le N}\left|\mathbf{n}\cdot\omega\right|^{2}\left\Vert u_{\mathbf{n}}\right\Vert^{2}=
\lim_{j\to\infty}\sum_{\left|\mathbf{n}\right|\le N}\left|\mathbf{n}\cdot\omega\right|^{2}\left\Vert
u_{j,\mathbf{n}}\right\Vert^{2}\le \\ \le
\liminf_{j\to\infty}\sum_{\mathbf{n}\in\mathbb{Z}^{k}}\left|\mathbf{n}\cdot\omega\right|^{2}\left\Vert
u_{j,\mathbf{n}}\right\Vert^{2}\le K\quad\forall N\in\mathbb{N}.\end{gather*}
 Hence, $u(\cdot)\in\mathrm{H}_{\omega}^{1}(\mathbb{T}^{k}\!\mapsto\!\mathbb{E}^{n})$
and $\left\Vert D_{\omega}u\right\Vert_{0}\le\liminf_{j\to\infty}\left\Vert D_{\omega}u_{j}\right\Vert_{0}.$
Moreover, $u_{j}(\cdot)$ converges to $u(\cdot)$ weakly in $\mathrm{H}_{\omega}^{1}(\mathbb{T}^{k}\!\mapsto\!\mathbb{E}^{n})$.

Next, for any bounded subset $\mathcal{A}\subseteq\mathcal{M}$ define
a functional space $\mathcal{H}_{\mathcal{A}}$ in a following way:
$u(\cdot)\in\mathcal{H}_{\mathcal{A}}$ iff there exists a sequence
$u_{j}(\cdot)\in\mathcal{S}_{\mathcal{A}}$ bounded in $\mathrm{H}_{\omega}^{1}(\mathbb{T}^{k}\!\mapsto\!\mathbb{E}^{n})$
and convergent to $u(\cdot)$ by norm of the space $\mathrm{H}(\mathbb{T}^{k}\!\mapsto\!\mathbb{E}^{n})$
(recall that we consider the set $\mathcal{A}\subseteq\mathcal{M}$
both as a subset of $\mathbb{E}^{n}$). As it was noted above $\mathcal{H}_{\mathcal{A}}\subset\mathrm{H}_{\omega}^{1}(\mathbb{T}^{k}\!\mapsto\!\mathbb{E}^{n})$.
We shall say that $h(\cdot)\in\mathrm{H}_{\omega}^{1}(\mathbb{T}^{k}\!\mapsto\!\mathbb{E}^{n})$
is a vector field along the map $u(\cdot)\in\mathcal{H}_{\mathcal{A}}$
defined in the above sens by a sequence $u_{j}(\cdot)$ if there exists
a sequence $h_{j}(\cdot)\in\mathrm{C}^{\infty}\left(\mathbb{T}^{k}\!\mapsto\! T\mathcal{M}\right)$
such that $h_{j}(\varphi)\in T_{u_{j}(\varphi)}\mathcal{M}$, the
sequences $\max_{\varphi\in\mathbb{T}^{k}}\left\Vert h_{j}(\varphi)\right\Vert $,
$\left\Vert h_{j}\right\Vert_{1}$ are bounded, and $\lim_{j\to\infty}\left\Vert h-h_{j}\right\Vert_{1}=0.$

\begin{defn}
A Besicovitch quasiperiodic function $u(\omega t)$ generated by a
function $u(\cdot)\in\mathcal{H}_{\mathcal{A}}$ is called a weak
quasiperiodic solution of the natural system on $\mathcal{M}$ if
it satisfies the equality \begin{equation}
\langle D_{\omega}u(\varphi),D_{\omega}h(\varphi)\rangle_{0}+\langle W_{x}^{\prime}(\varphi,u(\varphi)),h(\varphi)\rangle_{0}=0\label{eq:defgensol}\end{equation}
 for any vector field $h(\cdot)$ along $u(\cdot)$.
\end{defn}
This definition is natural since the equality \eqref{eq:defgensol}
holds true for any classical quasiperiodic solution $u(\omega t)$
and continuous vector field $h(\varphi)$ along $u(\cdot)$ with continuous
derivative $D_{\omega}h(\cdot)$. It should be also noted the following
fact.

The application of variational approach to the problem of detecting
weak quasiperiodic solution consists in finding a function $u_{\ast}(\cdot)\in\mathcal{H}_{\mathcal{A}}$
which takes values in appropriately chosen bounded subset $\mathcal{A}\subset\mathcal{M}$
and which is a strong limit in $\mathrm{H}(\mathbb{T}^{k}\mapsto\mathbb{E}^{n})$
of minimizing sequence for the functional (the averaged Lagrangian)
\begin{equation}
J[u]=\int_{\mathbb{T}^{k}}\left[\frac{1}{2}\|D_{\omega}u(\varphi)\|^{2}+W(\varphi,u(\varphi))\right]d\varphi\label{eq:mainfunctional}\end{equation}
restricted to $\mathcal{S}_{\mathcal{A}}$. It is naturally to expect
that the first variation of $J$ at $u_{\ast}(\cdot)$ vanishes, i.e.\begin{gather}
J^{\prime}[u_{\ast}]\left(h\right):=\langle D_{\omega}u_{\ast}(\varphi),D_{\omega}h(\varphi)\rangle_{0}+\langle W_{x}^{\prime}(\varphi,u_{\ast}(\varphi)),h(\varphi)\rangle_{0}=0\label{eq:varJeq0}\end{gather}
 for any vector field $h(\cdot)$ along $u_{\ast}(\cdot)$. In such
a case $u_{\ast}(\omega t)$ is a weak quasiperiodic solution.

In order to guarantee the convergence of a minimizing sequence $u_{j}(\cdot)\in\mathcal{S}_{\mathcal{A}}$
for $J\bigl|_{\mathcal{S}_{\mathcal{A}}}$ by norm $\left\Vert \cdot\right\Vert_{0}$
it is naturally to impose some convexity conditions both on the set
$\mathcal{A}$ and on the functional $J$. Usually, such conditions
are formulated by means of geodesics. But in the case where $\left(\mathcal{M},\left\langle \cdot,\cdot\right\rangle \right)$
is not a Riemannian manifold of non-positive sectional curvature,
we are not able to determine whether the functional of averaged kinetic
energy, namely $J_{1}[u]:=\frac{1}{2}\int_{\mathbb{T}^{k}}\|D_{\omega}u(\varphi)\|^{2}d\varphi$,
is convex using geodesics of Levi-Civita connection $\nabla$. if
$\left(\mathcal{M},\left\langle \cdot,\cdot\right\rangle \right)$.
(The case of Riemannian manifold of non-positive sectional curvature
was considered in \cite{ZakPar99A,ZakPar99B}.)

In order to overcome the above difficulty we introduce a conformally
equivalent inner product of the form $\left\langle \cdot,\cdot\right\rangle_{V}\bigl|_{T_{x}\mathcal{M}}:=\mathrm{e}^{V(x)}\left\langle \cdot,\cdot\right\rangle \bigl|_{T_{x}\mathcal{M}}$
with appropriately chosen smooth function $V(\cdot):\mathcal{M}\!\mapsto\!\mathbb{R}$.
With this approach we succeed in establishing a required convexity
properties of averaged Lagrangian under certain convexity conditions
imposed on functions $V(\cdot)$ and $W(\varphi,\cdot)$.

\medskip{}

\textbf{3. Convexity of averaged Lagrangian.}It is easily seen that
if $V(\cdot)\in\mathrm{C}^{\infty}(\mathcal{M}\!\mapsto\!\mathbb{R})$
is a bounded function on $\mathcal{M}$ then the Riemannian manifold
$\left(\mathcal{M},\left\langle \cdot,\cdot\right\rangle_{V}\right)$
equipped with corresponding Levi-Civita connection is complete. In
fact, by definition, the standard distance between any two points
$x_{1},x_{2}\in\left(\mathcal{M},\left\langle \cdot,\cdot\right\rangle \right)$
is defined as \begin{gather*}
\rho(x_{1},x_{2}):=\inf\left\{ l(c):c\in\mathcal{C}_{x_{1},x_{2}}\right\},\end{gather*}
 where $\mathcal{C}_{x_{1},x_{2}}$ is the set of all piecewise differentiable
paths $c:[0,1]\mapsto\mathcal{M}$ connecting $x_{1}$ with $x_{2}$,
and $l(c)$ is the length of $c$ on $\left(\mathcal{M},\left\langle \cdot,\cdot\right\rangle \right)$.
If we denote by $l_{V}(c)$ the length of path $c$ on $\left(\mathcal{M},\left\langle \cdot,\cdot\right\rangle_{V}\right)$,
then \begin{gather*}
\inf_{x\in\mathcal{M}}\sqrt{\mathrm{e}^{V(x)}}l(c)\le l_{V}(c)\le\sup_{x\in\mathcal{M}}\sqrt{\mathrm{e}^{V(x)}}l(c).\end{gather*}
 Hence, the metric $\rho(\cdot,\cdot)$ and the metric $\rho_{V}(\cdot,\cdot)$
of $\left(\mathcal{M},\left\langle \cdot,\cdot\right\rangle_{V}\right)$
are equivalent. Now it remains only to apply the Hopf–Rinow theorem
(see, e.g., \cite[sect. 5.3]{GKM71}).

In order to distinguish geodesics of metrics $\rho$ and $\rho_{V}$
we shall call them $\rho$-geodesic and $\rho_{V}$-geodesic respectively.

For $x\in\mathcal{M}$, let $\exp_{x}(\cdot):T_{x}\mathcal{M}$$\mapsto\mathcal{M}$
denotes the exponential mapping of Riemannian manifold $\left(\mathcal{M},\left\langle \cdot,\cdot\right\rangle \right)$
with Levi-Civita connection $\nabla$ and let $\exp_{x}^{V}(\cdot):T_{x}\mathcal{M}\!\mapsto\!\mathcal{M}$
be the analogous mapping of Riemannian manifold $\left(\mathcal{M},\left\langle \cdot,\cdot\right\rangle_{V}\right)$
with corresponding Levi-Civita connection $\nabla^{V}$. Note that
since both manifolds are complete the domains of both exponential
mappings coincide with entire $T_{x}\mathcal{M}$.

Recall that a set of a Riemannian manifold is called convex if together
with any two points $x_{1},x_{2}$ this set contains a (unique) minimal
geodesic segment connecting $x_{1}$ with $x_{2}$(see, e.g., \cite[sect. 11.8]{BisCrit64}
or \cite[sect. 5.2]{GKM71}). It is well known that for any point
$x_{0}$ an open ball of sufficiently small radius centered at point
$x_{0}$ is convex. A function $f:\mathcal{D}_{f}\!\mapsto\!\mathbb{R}$
with convex domain $\mathcal{D}_{f}\subset\mathcal{M}$ is convex
iff its superposition with any naturally parametrized geodesic containing
in $\mathcal{D}_{f}$ is convex.

Recall also that for the function $V(\cdot)$, the Hesse form $H_{V}(x)$
at point $x$ (see., e.g., \cite{GKM71}) is defined by the equality
\begin{gather*}
\left\langle H_{V}(x)\xi,\eta\right\rangle :=\left\langle \nabla_{\xi}\nabla V(x),\eta\right\rangle
\quad\forall\xi,\eta\in T_{x}\mathcal{M}.\end{gather*} In addition, let us introduce the
following quadratic form \begin{gather*} \left\langle G_{V}(x)\xi,\xi\right\rangle :=\left\langle
H_{V}(x)\xi,\xi\right\rangle -\frac{1}{2}\left\langle \nabla V(x),\xi\right\rangle^{2}\quad\forall\xi\in
T_{x}\mathcal{M},\end{gather*} and denote \begin{gather*} \lambda_{V}(x):=\min_{\xi\in
T_{x}\mathcal{M}\setminus\{0\}}\left\langle H_{V}(x)\xi,\xi\right\rangle /\left\Vert \xi\right\Vert^{2},\\
\mu_{V}(x):=\min_{\xi\in T_{x}\mathcal{M}\setminus\{0\}}\left\langle G_{V}(x)\xi,\xi\right\rangle /\left\Vert
\xi\right\Vert^{2}.\end{gather*}

We accept the following hypotheses concerning convexity properties
of functions $V(\cdot)$ and $W(\cdot)$:

\begin{description}
\item [{{(H1)}}] there exist a bounded function $V(\cdot)\in\mathrm{C}^{\infty}(\mathcal{M}\!\mapsto\!\mathbb{R})$
and a bounded domain $\mathcal{D}\subset\mathcal{M}$ such that\begin{gather}
\lambda_{V}(x)+\frac{1}{2}\left\Vert \nabla V(x)\right\Vert^{2}\ge0,\quad\forall x\in\mathcal{D};\label{eq:lambdV}\end{gather}

\item [{{(H2)}}] there exist a noncritical value $v\in V(\mathcal{D})$
and a connected component $\Omega$ of open sublevel set $V^{-1}((-\infty,v))$
with the following properties: (a) for any $x,y\in\Omega$ the domain
$\mathcal{D}$ contains a unique minimal $\rho_{V}$-geodesic segment
with endpoints $x,y$; (b) the second fundamental form of $\partial\Omega$
is positive at each point $x\in\partial\Omega$ (i.e. for any $x\in\partial\Omega$
the restriction of $H_{V}(x)$ to $T_{x}\partial\Omega$ is positive
definite); (c) the function $V(\cdot)$ satisfies the inequality \begin{gather}
\mu_{V}(x)\ge2K^{\ast}(x)\quad\forall x\in\Omega\label{eq:mugeK}\end{gather}
 where \begin{gather*}
K^{\ast}(x):=\max_{\sigma_{x}(\xi,\eta)}\frac{\left\langle R(\eta,\xi)\xi,\eta\right\rangle}{\left\Vert \eta\right\Vert^{2}\left\Vert \xi\right\Vert^{2}-\left\langle \eta,\xi\right\rangle^{2}}\end{gather*}
 is the maximum sectional curvature at point $x$, $R$ is the Riemann
curvature tensor of $\left(\mathcal{M},\left\langle \cdot,\cdot\right\rangle \right)$,
$\sigma_{x}(\xi,\eta)$ is a plane defined by vectors $\xi,\eta\in T_{x}\mathcal{M}$,
and $K(\sigma_{x}(\xi,\eta))$ is a sectional curvature in direction
$\sigma_{x}(\xi,\eta)$ \cite{GKM71};
\item [{{(H3)}}] the function $W(\cdot,\cdot)$ satisfies the following
inequalities \begin{gather*}
\lambda_{W}(\varphi,x)+\frac{1}{2}\left\langle \nabla W(\varphi,x),\nabla V(x)\right\rangle >0\quad\forall(\varphi,x)\in\mathbb{T}^{k}\times\bar{\Omega}\quad(\bar{\Omega}:=\Omega\cup\partial\Omega),\\
\left\langle \nabla W(\varphi,x),\nabla V(x)\right\rangle >0\quad\forall(\varphi,x)\in\mathbb{T}^{k}\times\partial\Omega\end{gather*}
where $\lambda_{W}(\varphi,x)$ is minimal eigenvalue of Hesse form
$H_{W}(\varphi,x)$ for the function $W(\varphi,\cdot):\mathcal{M}\!\mapsto\!\mathbb{R}$.
\end{description}
\begin{thm}
\label{thm:Jprop}Let the Hypotheses (H1)--(H3) hold true. Then there exist
positive constants $C$, $C_{1}$ and $c$ such that for any
$u_{0}(\cdot),u_{1}(\cdot)\in\mathrm{C}^{\infty}\left(\mathbb{T}^{k}\!\mapsto\!\Omega\right)$ one can choose a
vector field $h(\cdot)\in\mathrm{C}^{\infty}\left(\mathbb{T}^{k}\!\mapsto\! T\mathcal{M}\right)$ along
$u_{0}(\cdot)$ (this implies that $h(\varphi)\in T_{u_{0}(\varphi)}\mathcal{M}$ for all
$\varphi\in\mathbb{T}^{k}$) in such a way that the following inequalities hold true
\begin{gather*} c\rho(u_{0}(\varphi),u_{1}(\varphi))\le\left\Vert {h(\varphi)}\right\Vert
\le{C\rho(u_{0}(\varphi),u_{1}(\varphi))}\quad\forall\varphi\in{\mathbb{T}^{k}},\\ \left\Vert
D_{\omega}h(\varphi)\right\Vert \le C_{1}\left[\left\Vert D_{\omega}u_{0}(\varphi)\right\Vert +\left\Vert
D_{\omega}u_{1}(\varphi)\right\Vert \right]\quad\forall\varphi\in\mathbb{T}^{k},\\
J[u_{1}]-J[u_{0}]-J^{\prime}[u_{0}](h)\ge\frac{\varkappa
c^{2}}{2}\intop_{\mathbb{T}^{k}}\rho^{2}(u_{0},u_{1})\mathrm{d}\varphi\end{gather*}
 where $\varkappa:=\min\left\{ \lambda_{W}(\varphi,x)+\frac{1}{2}\left\langle \nabla W(\varphi,x),\nabla V(x)\right\rangle
 :(\varphi,x)\in\mathbb{T}^{k}\times\bar{\Omega}\right\}.$
\end{thm}
The proof of this theorem needs several auxiliary statements and will
be given below at the end of present Section.

\begin{prop}
\label{pro:Vgeoeq}The Euler-Lagrange equation for $\rho_{V}$-geodesic
on Riemannian manifold $\left(\mathcal{M},\left\langle \cdot,\cdot\right\rangle \right)$
has the form\begin{gather}
\nabla_{\dot{x}}\dot{x}=-\left\langle \nabla V(x),\dot{x}\right\rangle \dot{x}+\frac{\left\Vert \dot{x}\right\Vert^{2}}{2}\nabla V(x),\label{eq:Vgeoequ}\end{gather}

\end{prop}
\begin{proof}
A $\rho_{V}$-geodesic segment with endpoints $x_{0},x_{1}\in\mathcal{M}$
is an extremal of functional $\Phi[x(\cdot)]=\intop_{0}^{1}\mathrm{e}^{V\circ x(t)}\left\Vert \dot{x}(t)\right\Vert^{2}\mathrm{d}t$
defined on the space $\mathcal{C}_{x_{0}x_{1}}^{2}$of twice continuous
differentiable curves $x=x(t),$ $t\in[0,1]$, such that $x(0)=x_{0}$,
$x(1)=x_{1}$. We are going to derive the Euler-Lagrange equation
using the connection $\nabla$. Consider a variation of $x(\cdot)$
defined by a smooth mapping $y(\cdot,\cdot):[0,1]\times(-\varepsilon,\varepsilon)\mapsto\mathcal{M}$
such that $y(\cdot,\lambda)\in\mathcal{C}_{x_{0}x_{1}}^{\infty}$for
any fixed $\lambda\in(-\varepsilon,\varepsilon)$ and $y(t,0)\equiv x(t)$.
Put \[
\dot{y}(t,\lambda):=\frac{\partial}{\partial t}y(t,\lambda),\quad y^{\prime}(t,\lambda):=\frac{\partial}{\partial\lambda}y(t,\lambda).\]
 Obviously, $\dot{y}(t,0)=\dot{x}(t)$, $y(i,\lambda)\equiv x_{i}$,
and $y^{\prime}(i,\lambda)=0$, $i=0,1$. Then since
$\nabla_{y^{\prime}}\dot{y}=\nabla_{\dot{y}}y^{\prime}$, we have \begin{gather*}
\frac{\mathrm{d}}{\mathrm{d}\lambda}\bigl|_{\lambda=0}\intop_{0}^{1}\mathrm{e}^{V\circ y}\left\Vert
\dot{y}\right\Vert^{2}\mathrm{d}s=\\ =\intop_{0}^{1}\left[\mathrm{e}^{V\circ y}\left\langle \nabla V\circ
y,y^{\prime}\right\rangle \left\Vert \dot{y}\right\Vert^{2}+2\mathrm{e}^{V\circ y}\left\langle
\nabla_{y^{\prime}}\dot{y},\dot{y}\right\rangle \right]_{\lambda=0}\mathrm{d}t=\\
=\intop_{0}^{1}\left[\mathrm{e}^{V\circ y}\left\langle \nabla V\circ y,y^{\prime}\right\rangle \left\Vert
\dot{y}\right\Vert^{2}+2\mathrm{e}^{V\circ y}\left\langle \nabla_{\dot{y}}y^{\prime},\dot{y}\right\rangle
\right]_{\lambda=0}\mathrm{d}t.\end{gather*}
 Taking into account that \begin{gather*}
\frac{\partial}{\partial t}\mathrm{e}^{V\circ y}\left\langle y^{\prime},\dot{y}\right\rangle =\mathrm{e}^{V\circ y}\left\langle \nabla V\circ y,\dot{y}\right\rangle \left\langle y^{\prime},\dot{y}\right\rangle +\mathrm{e}^{V\circ y}\left\langle \nabla_{\dot{y}}y^{\prime},\dot{y}\right\rangle +\mathrm{e}^{V\circ y}\left\langle y^{\prime},\nabla_{\dot{y}}\dot{y}\right\rangle \end{gather*}
 and $\mathrm{e}^{V\circ y}\left\langle y^{\prime},\dot{y}\right\rangle \bigl|_{t=0,1}=0$,
we get \begin{gather*}
\intop_{0}^{1}\mathrm{e}^{V\circ y}\left\langle \nabla_{\dot{y}}y^{\prime},\dot{y}\right\rangle \mathrm{d}t=-\intop_{0}^{1}\mathrm{e}^{V\circ y}\left[\left\langle \nabla V\circ y,\dot{y}\right\rangle \left\langle y^{\prime},\dot{y}\right\rangle +\left\langle y^{\prime},\nabla_{\dot{y}}\dot{y}\right\rangle \right]\mathrm{d}t.\end{gather*}
 From this it follows that the first variation on functional $\Phi$
is \begin{gather*}
\frac{\mathrm{d}}{\mathrm{d}\lambda}\Bigl|_{\lambda=0}\Phi[y(\cdot,\lambda)]=\Phi^{\prime}[x(\cdot)]\left(y^{\prime}(\cdot,0)\right)=\\
=\intop_{0}^{1}\left[\mathrm{e}^{V}\left(\left\langle \nabla V,y^{\prime}\right\rangle \left\Vert \dot{x}\right\Vert^{2}-2\left\langle \nabla V,\dot{x}\right\rangle \left\langle \dot{x},y^{\prime}\right\rangle -2\left\langle \nabla_{\dot{x}}\dot{x},y^{\prime}\right\rangle \right)\right]\bigl|_{x=x(t),\lambda=0}\mathrm{d}t,\end{gather*}
 and the Euler-Lagrange equation is exactly \eqref{eq:Vgeoequ}.
\end{proof}
\begin{prop}
\label{pro:rogeodOm}Let the Hypothesis (H1) holds true. If a $\rho_{V}$-geodesic
segment connecting points $x_{0},x_{1}$ of the set $\Omega$ belongs
to $\mathcal{D}$, then this segment belongs to $\Omega$.
\end{prop}
\begin{proof}
Let $x(\cdot)\in\mathcal{C}_{x_{0}x_{1}}^{2}$ satisfies \eqref{eq:Vgeoequ} and let
$x(t)\in\mathcal{D}$ for all $t\in[0,1]$. Then \begin{gather*}
\frac{\mathrm{d}^{2}}{\mathrm{d}t^{2}}\mathrm{e}^{V}\Bigl|_{x=x(t)}=\\ =\left[\mathrm{e}^{V}\left(\left\langle
\nabla_{\dot{x}}\nabla V,\dot{x}\right\rangle +\left\langle \nabla V,-\left\langle \nabla V,\dot{x}\right\rangle
\dot{x}+\left\Vert \dot{x}\right\Vert^{2}\nabla V/2\right\rangle +\left\langle \nabla
V,\dot{x}\right\rangle^{2}\right)\right]\bigl|_{x=x(t)}=\\ =\left[\mathrm{e}^{V}\left(\left\langle
\nabla_{\dot{x}}\nabla V,\dot{x}\right\rangle +\left\Vert \dot{x}\right\Vert^{2}\left\Vert \nabla
V\right\Vert^{2}/2\right)\right]\bigl|_{x=x(t)}\ge\\ \ge \left[\mathrm{e}^{V}\left\Vert
\dot{x}\right\Vert^{2}\left(\lambda_{V}+\left\Vert \nabla
V\right\Vert^{2}/2\right)\right]\bigl|_{x=x(t)}\ge0.\end{gather*}
 Hence, $\mathrm{e}^{V\circ x(\cdot)}$ is convex and this implies
$V\circ x(t)<v$ for all $t\in[0,1]$.
\end{proof}
\begin{prop}
\label{pro:uniqmingeo}Under the Hypotheses (H1)-(H2), the minimal
$\rho_{V}$-geodesic segment connecting any two points $x,y\in\Omega$
does not contain conjugate points.
\end{prop}
\begin{proof}
It is known (see. \cite[sect. 3.6]{GKM71}) that the sectional curvature in
direction $\sigma_{x}(\xi_{1},\xi_{2})$ on Riemannian manifold
$\left(\mathcal{M},\mathrm{e}^{V}\left\langle \cdot,\cdot\right\rangle \right)$ is represented in the form\begin{gather*}
K_{V}\left(\sigma_{x}(\xi_{1},\xi_{2})\right)=\mathrm{e}^{-V}K(\sigma_{x}(\xi_{1},\xi_{2}))-\\ -
\frac{\mathrm{e}^{-V}}{2}\sum_{i=1}^{2}\left[\left\langle H_{V}(x)\xi_{i},\xi_{i}\right\rangle
-\frac{1}{2}\left\langle \nabla V(x),\xi_{i}\right\rangle^{2}\right]-\frac{\mathrm{e}^{-V}}{4}\left\Vert
\nabla V(x)\right\Vert^{2}\end{gather*} where $\xi_{1},\xi_{2}$ is an orthonormal
basis of the plane $\sigma_{x}(\xi_{1},\xi_{2})$, and the inequality
\eqref{eq:mugeK} yields that this curvature is non-positive for any
$x\in\bar{\Omega}$. By the Morse--Schoenberg theorem any $\rho_{V}$-geodesic segment
containing in $\bar{\Omega}$ does not contain conjugate points.
\end{proof}
\begin{prop}
Under the Hypotheses (H1)-(H3) there exists a smooth mapping
$\zeta(\cdot,\cdot):\Omega\times\Omega\mapsto T\mathcal{M}$ such that $\zeta(x,y)\in T_{x}\mathcal{M}$ and\begin{gather}
\exp_{x}^{V}(\zeta(x,y))=y,\quad\mathrm{e}^{V(x)/2}\left\Vert \zeta(x,y)\right\Vert
=\rho_{V}(x,y),\\ \exp_{x}^{V}(t\zeta(x,y))\in\Omega\quad\forall
t\in[0,1].\label{eq:propzeta}\end{gather}

\end{prop}
\begin{proof}
It is known that if for some $\xi\in T_{x}\mathcal{M}$ a geodesic
segment $\exp_{x}^{V}(t\xi),\; t\in[0,1]$, does not contain conjugate
points then the mapping $\exp_{x}^{V}(\cdot)$ is local diffeomorphism
at any point $t\xi$, $t\in[0,1]$. Under the Hypothesis (H2) for
any $x,y\in\Omega$ there exists a unique $\zeta(x,y)$ which satisfies
conditions \eqref{eq:propzeta}. It follows from the implicit function
theorem that the mapping $\zeta(\cdot,\cdot):\Omega\times\Omega\mapsto T\mathcal{M}$
is smooth.
\end{proof}
If we define the mapping \begin{gather*}
\gamma_{V}(\cdot,\cdot,\cdot):[0,1]\times\Omega\times\Omega\!\mapsto\!\Omega,\quad\gamma_{V}(t,x,y):=\exp_{x}^{V}(t\zeta(x,y)),\end{gather*}
then for any $x,y\in\mathcal{D}$ the mapping $\gamma_{V}(\cdot,x,y):[0,1]\!\mapsto\!\mathcal{D}$
satisfies the equation \eqref{eq:Vgeoequ} together with boundary
conditions $\gamma_{V}(0,x,y)=x$, $\gamma_{V}(1,x,y)=y$. The following
scalar differential equation \begin{gather*}
\frac{\mathrm{d}\tau}{\mathrm{d}s}=\exp\left(V\circ\gamma_{V}(\tau,x,y)\right)\intop_{0}^{1}\exp\left(-V\circ\gamma_{V}(t,x,y)\right)\mathrm{d}t.\end{gather*}
 has a unique strictly monotonically increasing solution \begin{equation}
\tau(\cdot,x,y):[0,1]\mapsto[0,1],\quad\tau(0,x,y)=0,\quad\tau(1,x,y)=1.\label{eq:tau}\end{equation}
 By means of reparametrisation $t=\tau(s,x,y)$ we define a smooth
mapping\begin{gather*}
\chi(\cdot,\cdot,\cdot):[0,1]\times\Omega\times\Omega\!\mapsto\!\Omega,\quad\chi(s,x,y):=\gamma_{V}(\tau(s,x,y),x,y)\end{gather*}
 which plays an important role in subsequent reasoning. In \cite{ZakPar99}
$\chi(\cdot,\cdot,\cdot)$ is called the connecting mapping.

\begin{prop}
\label{pro:chiprop}For any $x,y\in\Omega$ the mapping $\chi(\cdot,x,y):[0,1]\!\mapsto\!\Omega$
satisfies the equation\begin{gather}
\nabla_{x^{\prime}}x^{\prime}=\frac{\left\Vert x^{\prime}\right\Vert^{2}}{2}\nabla V(x),\label{eq:chiequ}\end{gather}
 where $x^{\prime}=\frac{\mathrm{d}x}{\mathrm{d}s}$ and the boundary
conditions $\chi(0,x,y)=x$, $\chi(1,x,y)=y$.
\end{prop}
\begin{proof}
The boundary conditions follow from definition of $\gamma_{V}$ and \eqref{eq:tau}.
Let us show that \eqref{eq:chiequ} is obtained from \eqref{eq:Vgeoequ} after
the change of independent variable $t=\tau(s)$. In fact, let $\chi(s)=x\circ\tau(s)$.
Then \eqref{eq:Vgeoequ} takes the form\textit{\begin{gather*}
\frac{1}{\tau^{\prime}}\nabla_{\chi^{\prime}}\left(\frac{1}{\tau^{\prime}}\chi^{\prime}\right)=-\frac{1}{(\tau^{\prime})^{2}}\left\langle
\nabla V\circ\chi,\chi^{\prime}\right\rangle \chi^{\prime}+\frac{\left\Vert
\chi^{\prime}\right\Vert^{2}}{2(\tau^{\prime})^{2}}\nabla V\circ\chi,\end{gather*}}
or\begin{gather*}
-\frac{\tau^{\prime\prime}}{\tau^{\prime}}\chi^{\prime}+\nabla_{\chi^{\prime}}\chi^{\prime}=-\left[\frac{\mathrm{d}}{\mathrm{d}s}V\circ\chi\right]\chi^{\prime}+\frac{\left\Vert
\chi^{\prime}\right\Vert^{2}}{2}\nabla V\circ\chi.\end{gather*}
 From this it follows \eqref{eq:chiequ} since $\tau^{\prime\prime}/\tau^{\prime}=\left(V\circ\chi\right)^{\prime}$.
\end{proof}
\begin{prop}
\label{pro:convJ1}Let $u_{i}(\cdot)\in\mathcal{S}_{\Omega}$, $i=0,1$.
Then under the hypotheses (H1)-(H2) the following inequality is valid\begin{gather*}
\frac{\mathrm{d}^{2}}{\mathrm{d}s^{2}}\left\Vert D_{\omega}\chi\left(s,u_{0}(\varphi),u_{1}(\varphi)\right)\right\Vert^{2}\ge0\quad\forall s\in[0,1],\;\forall\varphi\in\mathbb{T}^{k}.\end{gather*}

\end{prop}
\begin{proof}
For any fixed $\varphi\in\mathbb{T}^{k}$ put \begin{gather*}
\eta(s,t):=\frac{\partial}{\partial t}\chi\left(s,u_{0}(\varphi+\omega t),u_{1}(\varphi+\omega t)\right)\equiv D_{\omega}\chi\left(s,u_{0}(\varphi+\omega t),u_{1}(\varphi+\omega t)\right),\\
\xi(s,t):=\frac{\partial}{\partial s}\chi\left(s,u_{0}(\varphi+\omega t),u_{1}(\varphi+\omega t)\right).\end{gather*}
 Then in view of the well known relations (see. e.g., \cite{GKM71},DNF84) \begin{gather*}
\nabla_{\eta}\xi=\nabla_{\xi}\eta,\quad\nabla_{\eta}\nabla_{\xi}\xi-\nabla_{\xi}\nabla_{\eta}\xi=R(\eta,\xi)\xi\end{gather*}
 and \eqref{eq:chiequ}, we have \textit{\begin{gather*}
\nabla_{\xi}^{2}\eta=\nabla_{\eta}\nabla_{\xi}\xi-R(\eta,\xi)\xi=\\ =\left\langle
\nabla_{\eta}\xi,\xi\right\rangle \nabla V\circ\chi+\frac{\left\Vert
\xi\right\Vert^{2}}{2}\nabla_{\eta}\nabla V\circ\chi-R(\eta,\xi)\xi\end{gather*}} and hence,
\begin{gather*} \frac{\mathrm{d}^{2}}{\mathrm{d}s^{2}}\left\Vert \eta\right\Vert^{2}=2\left[\left\langle
\nabla_{\xi}^{2}\eta,\eta\right\rangle +\left\Vert \nabla_{\xi}\eta\right\Vert^{2}\right]=\\ =2\left\Vert
\nabla_{\xi}\eta\right\Vert^{2}+2\left\langle \nabla_{\xi}\eta,\xi\right\rangle \left\langle \nabla V\circ\chi,\eta\right\rangle
+\\ +\left\Vert \xi\right\Vert^{2}\left\langle \nabla_{\eta}\nabla V\circ\chi,\eta\right\rangle -2\left\langle
R(\eta,\xi)\xi,\eta\right\rangle \ge\\ \ge2\left\Vert \nabla_{\xi}\eta\right\Vert^{2}-2\left\Vert
\nabla_{\xi}\eta\right\Vert \left\Vert \xi\right\Vert \left|\left\langle \nabla V\circ\chi,\eta\right\rangle \right|+\\ +
\left\Vert \xi\right\Vert^{2}\left\langle \nabla_{\eta}\nabla V\circ\chi,\eta\right\rangle -2K^{\ast}\circ\chi\left\Vert
\xi\right\Vert^{2}\left\Vert \eta\right\Vert^{2}.\end{gather*} \textit{\emph{Once the
Hypothesis (H2) holds true, we get \begin{gather*}
\frac{\mathrm{d}^{2}}{\mathrm{d}s^{2}}\left\Vert \eta\right\Vert^{2}\ge \\ \ge 2\left\Vert
\xi\right\Vert^{2}\left\Vert \eta\right\Vert^{2}\left[r^{2}- \left|\left\langle \nabla V\circ\chi,\mathbf{e}\right\rangle
\right|r+\frac{1}{2}\left\langle \nabla_{\mathbf{e}}\nabla V\circ\chi,\mathbf{e}\right\rangle
-K^{\ast}\circ\chi\right]\ge0\end{gather*}
 where $r:=\frac{\left\Vert \nabla_{\xi}\eta\right\Vert}{\left\Vert \xi\right\Vert \left\Vert \eta\right\Vert}$}}.
\end{proof}
Now we are in position to prove the Theorem~\ref{thm:Jprop}. Let
$u_{i}(\cdot)\in\mathcal{S}_{\Omega}$, $i=0,1$. By means of connecting
mapping we get the following representation\begin{gather}
J[\chi(s,u_{0},u_{1})]=J[u_{0}]+sJ^{\prime}[u_{0}]\left(\chi_{s}^{\prime}(0,u_{0},u_{1})\right)+\frac{s^{2}}{2}\frac{\mathrm{d}^{2}}{\mathrm{d}s^{2}}\Bigl|_{s=\theta}J\left[\chi\left(s,u_{0},u_{1}\right)\right]\label{eq:TailJ}\end{gather}
with some $\theta\in(0,1)$. To estimate from below the term with
second derivative we make use of Proposition~\ref{pro:convJ1} which
together with the Hypothesis (H3) implies

\begin{gather*}
\frac{\mathrm{d}^{2}}{\mathrm{d}s^{2}}\left[\frac{1}{2}\left\Vert D_{\omega}\chi\left(s,u_{0}(\varphi),u_{1}(\varphi)\right)\right\Vert^{2}+W\left(\varphi,\chi(s,u_{0},u_{1})\right)\right]\ge\\
\ge\frac{\mathrm{d}}{\mathrm{d}s}\left\langle \nabla W(\varphi,\chi),\chi_{s}^{\prime}\right\rangle =\left\langle \nabla_{\chi_{s}^{\prime}}\nabla W(\varphi,\chi),\chi_{s}^{\prime}\right\rangle +\left\langle \nabla W(\varphi,\chi),\nabla_{\chi_{s}^{\prime}}\chi_{s}^{\prime}\right\rangle =\\
=\left\langle \nabla_{\chi_{s}^{\prime}}\nabla W(\varphi,\chi),\chi_{s}^{\prime}\right\rangle +\frac{\left\Vert \chi_{s}^{\prime}\right\Vert^{2}}{2}\left\langle \nabla W(\varphi,\chi),\nabla V(\chi)\right\rangle \ge\varkappa\left\Vert \chi_{s}^{\prime}\right\Vert^{2}.\end{gather*}
 By the definition of $\chi$ we have \begin{gather*}
\chi_{s}^{\prime}\left(s,u_{0},u_{1}\right)=\tau^{\prime}(s)\dot{\gamma}_{V}\left(\tau(s),u_{0},u_{1}\right)=\\
=\exp\left(V\circ\gamma_{V}(\tau(s),u_{0},u_{1})\right)\intop_{0}^{1}\exp\left(-V\circ\gamma_{V}(t,u_{0},u_{1})\right)\mathrm{d}t\dot{\gamma}_{V}\left(\tau(s),u_{0},u_{1}\right).\end{gather*}
 Since $\gamma_{V}(t,x,y)$ is $\rho_{V}-$geodesic,{
}then $\exp\left(V\circ\gamma_{V}\right)\left\Vert \dot{\gamma}_{V}\right\Vert^{2}$
does not depend on $t$ and \[
\mathrm{e}^{V(x)/2}\left\Vert \dot{\gamma}_{V}(0,x,y)\right\Vert =\mathrm{e}^{V(x)/2}\left\Vert \zeta(x,y)\right\Vert =\rho_{V}(x,y).\]
 Hence\begin{gather*}
\left\Vert \chi_{s}^{\prime}\left(s,u_{0},u_{1}\right)\right\Vert^{2}=
\left[\intop_{0}^{1}\exp\left(-V\circ\gamma_{V}(t,u_{0},u_{1})\right)\mathrm{d}t\right]^{2}\times \\ \times
\exp\left(V\circ\gamma_{V}(\tau(s),u_{0},u_{1})\right)\rho_{V}^{2}(u_{0},u_{1}),\end{gather*}
 and \eqref{eq:propzeta} implies that there exist positive constants
$C$, $c$ dependent only on $V(\cdot)$ and $\Omega$ such that \begin{gather}
c\rho(u_{0},u_{1})\le\left\Vert \chi_{s}^{\prime}\left(s,u_{0},u_{1}\right)\right\Vert \le C\rho(u_{0},u_{1}).\label{eq:normchiprime}\end{gather}
 Define $h(\varphi):=\chi_{s}^{\prime}\left(0,u_{0}(\varphi),u_{2}(\varphi)\right)$.
Then \eqref{eq:TailJ} with $s=1$ yields \begin{gather*}
J[u_{1}]-J[u_{0}]-J^{\prime}[u_{0}]\left(\chi_{s}^{\prime}(0,u_{0},u_{1})\right)\ge\frac{\varkappa c^{2}}{2}\intop_{\mathbb{T}^{k}}\rho^{2}(u_{0},u_{1})\mathrm{d}\varphi.\end{gather*}
 Finally, since the set $\Omega$ is bounded and the mapping $\chi$
is smooth, there exists positive constant $C_{1}$ such that \begin{gather*}
\left\Vert D_{\omega}h(\varphi)\right\Vert \le C_{1}\left[\left\Vert D_{\omega}u_{0}(\varphi)\right\Vert +\left\Vert D_{\omega}u_{1}(\varphi)\right\Vert \right]\quad\forall\varphi\in\mathbb{T}^{k}.\end{gather*}
 The proof of Theorem~\ref{thm:Jprop} is complete.

\medskip{}

\textbf{4. Main existence theorem.} Now we proceed to the main result
of this paper.

\begin{thm}
\label{thm:existqrs}Let the Hypotheses (H1)--(H3) hold true. Then
the natural system on Riemannian manifold $\left(\mathcal{M},\left\langle \cdot,\cdot\right\rangle \right)$
with Lagrangian density $L=\frac{1}{2}\langle\dot{x},\dot{x}\rangle+W(\omega t,x)$
has a weak quasiperiodic solution.
\end{thm}
\begin{proof}
The proof will consist of three steps.

1. Construction of a projection mapping and its smooth approximation. Put
$\Omega+\delta=\left(\bigcup_{x\epsilon\Omega}B(x;\delta)\right)$ where $B(x;\delta)$ stands for an open ball
of radius $\delta$ centered at $x\in\mathcal{M}$ on Riemannian manifold $\left(\mathcal{M},\left\langle
\cdot,\cdot\right\rangle \right)$. Since by Hypothesis (H2) $v$ is a noncritical value, then
$\partial\Omega=V^{-1}(v)$ is a regular hypersurface with unit normal field
$\mathbf{\boldsymbol{\nu}:}=\frac{{\nabla}V}{\left\Vert {\nabla}V\right\Vert}$. As is well
known (see, e.g., \cite{BisCrit64}), for sufficiently small $\delta>0$, one can
correctly define the projection mapping $P_{\Omega}:\Omega+\delta\rightarrow\bar{\Omega}$ such
that $P_{\Omega}x\in\bar{\Omega}$ is the nearest point to $x\in\Omega+\delta$. If $x=X(q),$
$q\in\mathcal{Q}\subset\mathbb{R}^{m-1}$, is a smooth local parametric representation
of $\partial\Omega$ in a neighborhood of a point $x_{0}\in\partial\Omega$, then for sufficiently
small $\delta_{0}>0$ the mapping \[
\mathcal{Q}\times(-\delta_{0},\delta_{0})\ni(q,z)\mapsto\exp_{X(q)}\left(z\boldsymbol{\nu}\circ
X(q)\right)\]
 introduces local coordinates with the following properties: local
equation of $\partial\Omega$ is $z=0$; each naturally parametrized
$\rho$-geodesic $\gamma(s)=\exp_{X(q)}\left(s\boldsymbol{\nu}\circ X(q)\right)$
is orthogonal to each hypersurface $z=\mathrm{const}$; the Riemannian
metric takes the form $\sum_{i,j=1}^{m-1}b_{ij}(q,z)\mathrm{d}q_{i}\mathrm{d}q_{j}+\mathrm{d}z^{2}$,
where $B(q,z)=\left\{ b_{ij}(q,z)\right\}_{i,j=1}^{m-1}$ is positive
definite symmetric matrix; the function $V(\cdot)$ is represented
in the form $V(q,z)=v+a(q)z+b(q,z)z^{2}$; the mapping $P_{\Omega}$
has the form \begin{gather*}
P_{\Omega}(q,z):=\begin{cases}
(q,0) & \text{if}\quad z\in(0,\delta_{0}),\\
(q,z) & \text{if}\quad z\in(-\delta_{0},0].\end{cases}\end{gather*}
 The projection mapping is continuous on $\Omega+\delta$ and continuously
differentiable on $(\Omega+\delta)\backslash\partial\Omega$. Moreover,
it turns out that for sufficiently small $\delta>0$ the derivative
$P_{\Omega\ast}$ is contractive on $(\Omega+\delta)\backslash\partial\Omega$,
i.e.\begin{gather}
\left\Vert P_{\Omega\ast}\xi\right\Vert \le\left\Vert \xi\right\Vert \quad\forall\xi\in T_{x}\mathcal{M},\; x\in(\Omega+\delta)\backslash\partial\Omega.\label{eq:Pcontract}\end{gather}
 It is sufficiently to prove this inequality for any $x\in(\Omega+\delta)\backslash\partial\Omega$.
Let $q=q(s)$, $z=z(s)$ be natural equations of $\rho$-geodesic
which starts at a point $x_{0}=(q_{0},0)\in\partial\Omega$ in direction
of vector $\mathbf{\eta=}(\dot{q}_{0},0)\in T_{x_{0}}\partial\Omega$.
The hypothesis (H2) implies that \begin{gather*}
\left\langle \nabla_{\eta}\nabla V(x_{0}),\mathbf{\eta}\right\rangle =\frac{\mathrm{d}^{2}}{\mathrm{d}s^{2}}\Bigl|_{s=0}V(q(s),z(s))>0\quad\Leftrightarrow\quad a(q_{0})\ddot{z}(0)>0.\end{gather*}
 Since $a(q_{0})>0$ ($\boldsymbol{\nu}$ is external normal to $\partial\Omega$)
and $z$-component of geodesic equations is\begin{gather*}
\ddot{z}=\frac{1}{2}\frac{\partial}{\partial z}\sum_{i,j=1}^{m-1}b_{ij}(q,z)\dot{q}_{i}^{2}\dot{q}_{j}^{2},\end{gather*}
 then the matrix $B_{z}^{\prime}(q_{0},0)$ is positive definite.
From this it follows that $B(q,z_{1})>B(q,z_{2})$ for all $q$ from a
neighborhood of $q_{0}$ and all $z_{1},z_{2}\in(-\delta,\delta)$, $z_{1}>z_{2}$ if
$\delta\in(0,\delta_{0})$ is sufficiently small. Let $\xi=(\dot{q},\dot{z})$ be a tangent
vector at point $(q,z)$ where $z\in(0,\delta)$. Then \begin{gather*} \left\Vert
\xi\right\Vert^{2}=\sum_{i,j=1}^{m-1}b_{ij}(q,z)\dot{q}_{i}\dot{q}_{j}+\dot{z}^{2}\ge\\
ge  \sum_{i,j=1}^{m-1}
b_{ij}(q,z)\dot{q}_{i}\dot{q}_{j}\ge\sum_{i,j=1}^{m-1}b_{ij}(q,0)\dot{q}_{i}\dot{q}_{j}=\left\Vert
(\dot{q},0)\right\Vert^{2} =\left\Vert P_{\Omega\ast}\xi\right\Vert^2.\end{gather*}

Let us introduce a smooth approximation of projection mapping in a following
way. For $\varepsilon\in(0,\delta)$ define\begin{gather*}
\varpi_{\varepsilon}(z):=\begin{cases} \exp\left(1/z-1/(z+\varepsilon)\right), &
z\in(-\varepsilon,0),\\ 0, & z\in\mathbb{R}\setminus(-\varepsilon,0),\end{cases}\\
Z_{\varepsilon}(z):=\intop_{-\varepsilon}^{z}\frac{\intop_{s}^{0}
\varpi_{\varepsilon}(t)\mathrm{d}t}{\intop_{-\varepsilon}^{0}\varpi_{\varepsilon}(t)\mathrm{d}t}\mathrm{d}s-\varepsilon,
\quad z\in(-\delta_{0},\delta_{0})\end{gather*} Obviously that the function
$Z_{\varepsilon}(\cdot)$ is smooth, its derivative, $Z_{\varepsilon}^{\prime}(z)$,
equals 1 for $z\in(-\delta_{0},-\varepsilon]$, monotonically decreases from 1 to 0
on $[-\varepsilon,0]$, and equals 0 for $z\ge0$. From this it follows that
$Z_{\varepsilon}(z)$ equals $z$ for $z\in(-\delta_{0},-\varepsilon]$ monotonically
increases from $-\varepsilon$ to $Z_{\varepsilon}(0)\in(-\varepsilon,0)$ on
$[-\varepsilon,0]$, and equals $Z_{\varepsilon}(0)$ for $z\in[0,\delta_{0})$. Now
locally define \begin{gather*} P_{\varepsilon,\Omega}(q,z):=\begin{cases}
(q,Z_{\varepsilon}(0)) & \text{if}\quad z\in(0,\delta_{0}),\\
(q,Z_{\varepsilon}(z)) & \text{if}\quad z\in(-\delta_{0},0]\end{cases}\end{gather*}
and for each point $x\in\Omega$ such that $B(x;\delta)\subset\Omega$ put
$P_{\varepsilon,\Omega}(x)=x$. Since $Z_{\varepsilon}(0)<0$, then \[
P_{\varepsilon,\Omega}(\Omega+\delta)\subset\Omega\] and since
$\left|Z_{\varepsilon}^{\prime}(z)\right|\le1$, then for any $z\in(-\delta,\delta)$, and for
any tangent vector $\xi=(\dot{q},\dot{z})$ at point $(q,z)$ we have \begin{gather*}
\left\Vert
\xi\right\Vert^{2}=\sum_{i,j=1}^{m-1}b_{ij}(q,z)\dot{q}_{i}\dot{q}_{j}+\dot{z}^{2}\ge
\sum_{i,j=1}^{m-1}b_{ij}(q,Z_{\varepsilon}(z))\dot{q}_{i}\dot{q}_{j}+\left(Z_{\varepsilon}^{\prime}(z)\dot{z}\right)^{2}
=\\ =\left\Vert (\dot{q},Z_{\varepsilon}^{\prime}(z)\dot{z})\right\Vert^{2}=\left\Vert
P_{\varepsilon,\Omega\ast}\xi\right\Vert.\end{gather*}
 From this it follows that \begin{gather}
\left\Vert P_{\varepsilon,\Omega\ast}\xi\right\Vert \le\left\Vert \xi\right\Vert \quad\forall x\in\Omega+\delta,\;\forall\xi\in T_{x}\mathcal{M}.\label{eq:PepsOm}\end{gather}
 Besides, the Hypothesis (H3) implies

\begin{equation}
W(\varphi,P_{\varepsilon,\Omega}x)\leq W(\varphi,x)\quad\forall\varphi\in\mathbb{T}^{m},\,\forall
x\in\Omega+\delta\label{W}\end{equation} for sufficiently small $\delta$ and
$\varepsilon\in(0,\delta)$.

2. Minimization of functional $J$ on $\mathcal{S}_{\Omega+\delta}$.
Obviously that the functional $J$ restricted to $S_{\Omega+\delta}$
is bounded from below. Let us show that\begin{gather}
J_{\ast}:=\inf J[\mathcal{S}_{\Omega+\delta}]=\inf J[\mathcal{S}_{\Omega}].\label{eq:infJ}\end{gather}
In fact, if $v_{j}(\cdot)\in\mathcal{S}_{\Omega+\delta}$ is such
a sequence that $J[v_{j}]$ monotonically decreases to $J_{\ast}$,
then \eqref{eq:PepsOm} and \eqref{W} implies \begin{gather*}
J_{\ast}\le J[P_{\varepsilon/j,\Omega}v_{j}]\le J[v_{j}].\end{gather*}
Hence, the sequence $u_{j}(\cdot):=P_{\varepsilon/j,\Omega}v_{j}(\cdot)$
is minimizing both for $J\bigl|_{S_{\Omega}}$ and for $J\bigl|_{\mathcal{S}_{\Omega+\delta}}$.

3. Convergence of minimizing sequence to a weak solution. Let
$u_{j}(\cdot)\in\mathcal{S}_{\Omega}$ be a minimizing sequence for $J\bigl|_{\mathcal{S}_{\Omega}}$.
Without loss of generality, we may consider that \begin{gather} \left\Vert
D_{\omega}u_{j}\right\Vert_{0}^{2}\le
M:=\frac{2}{(2\pi )^k}\sup_{x\in\Omega}\intop_{\mathbb{T}^{k}}W(\varphi,x)\mathrm{d}\varphi-
\frac{2}{(2\pi )^k}\intop_{\mathbb{T}^{k}}\inf_{x\in\Omega}W(\varphi,x)\mathrm{d}\varphi.\label{eq:boundDuj}\end{gather}
Let $h_{j}(\cdot)\in\mathrm{C}^{\infty}\left(\mathbb{T}^{k}\!\mapsto\! T\mathcal{M}\right)$ be a sequence of
smooth mappings such that $h_{j}(\varphi)\in T_{u_{j}(\varphi)}\mathcal{M}$ for any
$\varphi\in\mathbb{T}^{k}$ and besides there exist positive constants $K,\; K_{1}$
such that \begin{gather} \left\Vert h_{j}\right\Vert_{1}\le K_{1},\quad\left\Vert
h_{j}(\varphi)\right\Vert \le K\quad\forall\varphi\in\mathbb{T}^{k}, \quad\forall
j=1,2,\ldots\label{eq:secvf}\end{gather} Let us show that \begin{gather}
\lim_{j\to\infty}J^{\prime}[u_{j}](h_{j})=0.\label{eq:limJ'}\end{gather}
 On one hand, $J[u_{j}]$ decreases to $J_{\ast}:=\inf J[\mathcal{S}_{\Omega}]$.
On the other hand, for sufficiently small $s_{0}\le1$ and for any
$j\in\mathbb{N}$ there exists a number $\theta_{j}\in[-s_{0},s_{0}]$ such that
\begin{gather*} J[\exp_{u_{j}}(sh_{j})]=J[u_{j}]+sJ^{\prime}[u_{j}](h_{j})+
\frac{s^{2}}{2}\frac{\mathrm{d}^{2}}{\mathrm{d}s^{2}}\Bigl|_{s=\theta_{j}}J[\exp_{u_{j}}(sh_{j})]\\
\forall s\in[-s_{0},s_{0}],\quad\forall j\in\mathbb{N},\end{gather*} and,
besides, there exists a constant $K_{2}>0$ such that \begin{gather*}
\left|\frac{\mathrm{d}^{2}}{\mathrm{d}s^{2}}J[\exp_{u_{j}}(sh_{j})]\right|\le K_{2}\quad\forall
s\in[-s_{0},s_{0}],\quad\forall j\in\mathbb{N}.\end{gather*}
 If now we suppose that $\limsup_{j\to\infty}\left|J^{\prime}[u_{j}](h_{j})\right|>0$
then one can choose $j$ and $s_{j}\in[-s_{0},s_{0}]$ in such a way
that \begin{gather*}
\exp_{u_{j}}(s_{j}h_{j})\in\mathcal{S}_{\Omega+\delta},\quad J[\exp_{u_{j}}(s_{j}h_{j})]<J_{\ast}.\end{gather*}
 Thus, in view of \eqref{eq:infJ}, we arrive at contradiction with
definition of $J_{\ast}$.

Now by Theorem~\ref{thm:Jprop} for any pair $u_{i+j}(\cdot),\; u_{j}(\cdot)$
there exists a vector field $h_{ij}(\cdot)$ along $u_{j}(\cdot)$
such that \begin{gather*}
J[u_{i+j}]-J[u_{j}]-J^{\prime}[u_{j}](h_{ij})\ge\frac{\varkappa c^{2}}{2}\intop_{\mathbb{T}^{k}}\rho^{2}(u_{j},u_{i+j})\mathrm{d}\varphi\ge
\\ \ge \frac{(2\pi)^{k}\varkappa c^{2}}{2}\left\Vert u_{i+j}-u_{j}\right\Vert_{0}^{2}.\end{gather*}
Since \eqref{eq:limJ'} implies $J^{\prime}[u_{j}](h_{ij})\to0$ as $j\to\infty$, then
the sequence $u_{j}(\cdot)$ is fundamental in
$\mathrm{H}(\mathbb{T}^{k}\!\mapsto\!\mathbb{E}^{n})$ and in view of \eqref{eq:boundDuj}
converges to a function $u_{\ast}(\cdot)$ strongly in
$\mathrm{H}(\mathbb{T}^{k}\!\mapsto\!\mathbb{E}^{n})$ and weakly in
$\mathrm{H}_{\omega}^{1}(\mathbb{T}^{k}\!\mapsto\!\mathbb{E}^{n})$. Without loss of
generality we may consider that $u_{\ast}(\cdot)$ is defined by a minimizing sequence
which converges a.e.

Now it remains only to prove that $u_{\ast}(\cdot)$ is a weak solution, i.e. that
there holds \eqref{eq:varJeq0}. Let $h(\cdot)$ be a vector field along $u_{\ast}(\cdot)$.
By definition, there exists a sequence of smooth mappings $h_{j}(\varphi)\in
T_{u_{j}(\varphi)}\mathcal{M}$ which satisfies \eqref{eq:secvf} and \eqref{eq:limJ'}. Then,
in view of \eqref{eq:boundDuj}, we get \begin{gather*} \lim_{j\to\infty}\left|\left\langle
D_{\omega}u_{\ast},D_{\omega}h\right\rangle_{0}-\left\langle D_{\omega}u_{j},D_{\omega}h_{j}\right\rangle_{0}\right|\le\\ \le
\lim_{j\to\infty}\left|\left\langle
D_{\omega}\left(u_{\ast}-u_{j}\right),D_{\omega}h\right\rangle_{0}\right|+\sqrt{M}\lim_{j\to\infty}\left\Vert
D_{\omega}\left(h-h_{j}\right)\right\Vert_{0}=0,\end{gather*}
 and by the Lebesgue theorem\begin{gather*}
\lim_{j\to\infty}\intop_{\mathbb{T}^{k}}\left[W\left(\varphi,u_{j}(\varphi)\right)-W\left(\varphi,u_{\ast}(\varphi)\right)\right]\mathrm{d}\varphi=0.\end{gather*}
Hence, \begin{gather*}
J^{\prime}[u_{\ast}](h)=\lim_{j\to\infty}J^{\prime}[u_{j}](h_{j})=0.\end{gather*}

\end{proof}

\end{document}